\documentclass[a4paper,10pt]{article}
\usepackage{stmaryrd}
\usepackage{amsfonts}
\usepackage{bbm}
\usepackage{amscd}
\usepackage{mathrsfs}
\usepackage{latexsym,amssymb,amsmath,amscd,amscd,amsthm,amsxtra,xypic}
\usepackage[dvips]{graphicx}
\usepackage[utf8]{inputenc}
\usepackage[T1]{fontenc}
\usepackage{lmodern}
\usepackage{amssymb}
\usepackage[all]{xy}
\usepackage{nicefrac,mathtools,enumitem}
\usepackage{microtype}
\usepackage{color}
\newtheorem{thm}{Theorem}[section]
\newtheorem{lem}[thm]{Lemma}
\newtheorem{cor}[thm]{Corollary}
\newtheorem{pro}[thm]{Proposition}
\newtheorem{ex}[thm]{Example}
\newtheorem{rmk}[thm]{Remark}
\newtheorem{defi}[thm]{Definition}

\newcommand{\be }{\begin{equation}}
\newcommand{\ee }{\end{equation}}

\newcommand{\pf}{\noindent{\bf Proof.}\ }
\newcommand {\emptycomment}[1]{}

\newcommand{\huaG}{\mathcal{G}}

\newcommand{\frka}{\mathfrak a}

\newcommand{\frkd}{\mathfrak d}

\newcommand{\frkg}{\mathfrak g}

\newcommand{\frkh}{\mathfrak h}

\newcommand{\frkk}{\mathfrak k}

\def\qed{\hfill ~\vrule height6pt width6pt depth0pt}

\newcommand{\half}{\frac{1}{2}}

\newcommand{\br}[1]{   [ \cdot,    \cdot  ]   }

\newcommand{\Id}{\rm{Id}}

\newcommand{\g}{\mathfrak g}

\newcommand{\dM}{\mathrm{d}}

\newcommand{\gl}{\mathfrak {gl}}

\newcommand{\ad}{\mathrm{ad}}
\newcommand{\add}{\frka\frkd}

\newcommand{\Img}{\mathrm{Im}}
\newcommand{\Cen}{\mathrm{Cen}}

\begin{document}
\title{
{A new approach to hom-Lie bialgebras
\thanks
 {
Research partially supported by NSFC  (11101179, 11271202, 11221091) and SRFDP
(20100061120096, 20120031110022).
 }
 \author{Yunhe Sheng \\
 Department of Mathematics, Jilin University,
 \\Changchun 130012, Jilin, China\\\vspace{2mm} Email: shengyh@jlu.edu.cn\\
  Chengming Bai\\
Chern Institute of Mathematics and LPMC, Nankai University,\\
Tianjin 300071, China \\ Email: baicm@nankai.edu.cn \\
}
} }

\date{}
\footnotetext{{\it{Keyword}:  hom-Lie algebras,  hom-Lie bialgebras,
Manin triples, $r$-matrices, $\mathcal O$-operators, hom-left-symmetric algebras }}

\footnotetext{{\it{MSC}}:  17A30,17B62.}

\maketitle
\begin{abstract}
In this paper, we introduce a new definition of a hom-Lie bialgebra, which is equivalent to a Manin triple of hom-Lie algebras. We also introduce a notion of an $\mathcal O$-operator and then construct solutions of the classical hom-Yang-Baxter equation in terms of $\mathcal O$-operators and hom-left-symmetric algebras.
\end{abstract}

\tableofcontents
\section{Introduction}

The notion of a hom-Lie algebra was introduced by Hartwig, Larsson,
and Silvestrov in \cite{HLS} as part of a study of deformations of
the Witt and the Virasoro algebras. In a hom-Lie algebra, the Jacobi
identity is twisted by a linear map, called the hom-Jacobi identity.
Some $q$-deformations of the Witt and the Virasoro algebras have the
structure of a hom-Lie algebra \cite{HLS}. Because of close relation
to discrete and deformed vector fields and differential calculus
\cite{HLS,LD1,LD2}, hom-Lie algebras are widely studied in the following aspects: representation and cohomology theory \cite{AMM,BM,shenghomLie,Yao2}, deformation theory \cite{MS1}, categorification theory \cite{shenghomLie2}, bialgebra theory \cite{LiangyunZhang,Yao1,Yao3}. See
\cite{MS2,MY,Yao4,BaiHom} for other interesting hom-algebra structures.

In \cite{Yao1}, the author introduced a notion of a hom-Lie bialgebra, and study coboundary ones and triangular ones in detail. A hom-Lie bialgebra is a hom-Lie algebra $(\g,[\cdot,\cdot]_\g,\phi_\g)$ together with a cobracket $\Delta:\g\longrightarrow\wedge^2\g$ such that $\Delta$ is a $1$-cocycle:
 $$
 \Delta[x,y]_\g=\ad_{\phi_\g(x)}\Delta(y)-\ad_{\phi_\frkg(y)}\Delta(x).
 $$
 In the case of the ordinary Lie algebras, there is a Manin triple of Lie algebras which is equivalent to a Lie bialgebra. However, a hom-Lie bialgebra defined above does not satisfy this property.

 In this paper, we introduce a new definition of a hom-Lie bialgebra (Definition \ref{defi:homLiebi}). Now given a hom-Lie bialgebra $(\g,\g^*)$, $\g\oplus\g^*$ is a hom-Lie algebra such that $(\g\oplus\g^*;\g,\g^*)$ is a Manin triple of hom-Lie algebras. We also study the classical hom-Yang-Baxter equation in detail, and construct an $r$-matrix in the semidirect hom-Lie algebra by introducing a notion of an $\mathcal O$-operator for a hom-Lie algebra.

 The paper is organized as follows. In Section 2, we recall  representations of hom-Lie algebras and corresponding coboundary operators. In particular, we extend the hom-Lie bracket to $\Lambda^\bullet\g$, which will be used when we consider coboundary hom-Lie bialgebras. Given a representation $(V,A,\rho)$, we also give the condition under which $(V^*,A^*,\rho^*)$ is also a representation. In this case, we say that the representation is admissible. We also give some basic properties for  admissible representations.
 In Section 3, we introduce the notions of a matched pair of hom-Lie algebras, a hom-Lie bialgebra, and a Manin triple of hom-Lie algebras. The following three expressions are equivalent: $(\g,\g^*)$ is a hom-Lie bialgebra, $(\g,\g^*;\ad^*,\add^*)$ is a matched pair of hom-Lie algebras, and $(\g\oplus\g^*;\g,\g^*)$ is a standard Manin triple of hom-Lie algebras. In Section 4, we study coboundary hom-Lie bialgebras and triangular hom-Lie bialgebras. In particular, we describe the condition of $r$ being a solution of the classical hom-Yang-Baxter equation using a cocycle condition. We also consider Lagrangian hom-subalgebras at the end of this section: we prove that $\huaG_R$ is a Lagrangian hom-subalgebra of $\g\oplus\g^*$, which is the double of a hom-Lie bialgebra $(\g,\g^*)$ if and only if $R$ satisfies a Maurer-Cartan type equation (Theorem \ref{thm:R}). In Section 5, we  study $r$-matrices in terms of operator forms. We introduce a notion of an $\mathcal O$-operator, which can be viewed as a special hom-Nijenhuis operator for the semidirect product hom-Lie algebra. The relations between $\mathcal O$-operators and hom-left-symmetric algebras are studied. At last, we construct  solutions of the classical hom-Yang-Baxter equation in the semidirect product hom-Lie algebras in terms of $\mathcal O$-operators (Theorem \ref{thm:O-operator}) and hom-left-symmetric algebras (Corollary \ref{co:hlsa}) .

Throughout this paper, all vector spaces and algebras are assumed to be finite-dimensional, although many results still hold in the
infinite-dimensional case.

\vspace{3mm}
{\bf Notations:} $x,~y,~z,~x_i,~y_i$ are elements in $\g$, $\xi,~\eta,~\gamma$ are elements in $\g^*$ and $u,v,w$ are elements in $V$. $\ad$ and $\add$ are the adjoint representation associated to hom-Lie algebras $\g$ and $\g^*$ respectively.

\section{Representations of hom-Lie algebras}

In this section, we recall some basic notions and facts about
hom-Lie algebras and their representations
\cite{HLS,shenghomLie}.\vspace{1mm}

\begin{defi}{\rm\cite{HLS}}
  A {\bf hom-Lie algebra}\footnote{The hom-Lie algebras defined here are also called multiplicative hom-Lie algebras in some references.} is a triple $(\frkg,\br__\frkg ,\phi_\frkg)$ consisting of a
  linear space $\frkg$, a skew-symmetric bilinear map (bracket) $\br,_\frkg:\wedge^2\frkg\longrightarrow
  \frkg$ and an algebra homomorphism $\phi_\frkg:\frkg\longrightarrow\frkg$ satisfying
  \begin{equation}
    [\phi_\frkg(x),[y,z]_\frkg]_\frkg+[\phi_\frkg(y),[z,x]_\frkg]_\frkg+[\phi_\frkg(z),[x,y]_\frkg]_\frkg=0.
  \end{equation}
A hom-Lie algebra $(\frkg,\br__\frkg ,\phi_\frkg)$ is said to be
{\bf regular (involutive)}, if $\phi_\frkg$ is nondegenerate (satisfies
$\phi_\frkg^2=\Id$).

A subspace $\frkh\subset\frkg$ is a {\bf hom-Lie subalgebra} of
$(\frkg,\br ,_\g,\phi_\g)$ if
 $\phi_\g(\frkh)\subset\frkh$ and
  $\frkh$ is closed under the bracket operation $\br,_\g$, i.e.
  $$
[x,x^\prime]_\frkg\in\frkh,\quad\forall~x,x^\prime \in\frkh.
  $$
\end{defi}
The bracket operation could be extended to $\Lambda^\bullet\frkg$
via
$$
[x_1\wedge\cdots x_m,y_1\wedge \cdots
y_n]_\g=\sum_{i,j}(-1)^{i+j}[x_i,y_j]_\g \wedge
\phi_\g(x_1)\wedge\cdots \widehat{x_i}\cdots \wedge
\phi_\g(x_m)\wedge \phi_\g(y_1)\wedge\cdots
\widehat{y_j}\cdots\wedge \phi_\g(y_n).
$$
By the fact that $\phi_\g$ is an algebra homomorphism, it is easy to see that the extended bracket satisfies
\begin{eqnarray*}
~[P,Q]_\g&=&-(-1)^{(p-1)(q-1)} [Q,P]_\g,\\
~  [P,Q\wedge
R]_\g&=&[P,Q]_\g\wedge\phi_\g(R)+(-1)^{q(p-1)}\phi_\g(Q)\wedge[P,R]_\g,\\
~[\phi_\g(P),[Q,R]_\g]_\g&=&[[P,Q]_\g,\phi_\g(R)]_\g+(-1)^{(p-1)(q-1)}[\phi_\g(Q),[P,R]_\g]_\g,
  \end{eqnarray*}
  for any $P\in\Lambda^p\g,~Q\in\Lambda^q\g,~R\in\Lambda^r\g$.

  For any $x\in\g$, define $\ad_x:\Lambda^p\g\longrightarrow\Lambda^p\g$ by $\ad_xP=[x,P]_\g$. Its dual map $\ad^*_x:\Lambda^p\g^*\longrightarrow \Lambda^p\g^*$ is defined as usual:$\langle\ad_x^*\Phi,P\rangle=-\langle\Phi,\ad_xP\rangle$. More precisely, for any $\xi_1,\cdots,\xi_p\in\g^*$, we have
  \begin{equation}
    \ad^*_x(\xi_1\wedge\cdots\xi_p)=\sum_{i=1}^p\phi_\g^*(\xi_1)\wedge\cdots \ad^*_x\xi_i\wedge\cdots\phi^*_\g(\xi_p).
  \end{equation}

\begin{defi} A {\bf homomorphism of hom-Lie algebras} $f:(\frkg,\br__\frkg
,,\phi_\frkg)\longrightarrow(\frkk,\br__\frkk,\phi_\frkk)$ is a
linear map $f:\frkg\longrightarrow\frkk$ such that
\begin{eqnarray}
\label{eqn:phimorphism1}f[x,y]_\frkg&=&[f(x),f(y)]_\frkk,\\\label{eqn:phimorphism2}
f\circ\phi_\frkg&=&\phi_\frkk\circ f.
\end{eqnarray}
\end{defi}

Let $(\frkg,\br__\frkg,\phi_\frkg)$ be a hom-Lie algebra and $V$ be an
arbitrary vector space. Let $A\in\gl(V)$ be an arbitrary linear
transformation from $V$ to $V$. The notion of a representation of a hom-Lie
algebra $(\frkg,\br__\frkg,\phi_\frkg)$ was introduced in \cite{shenghomLie}.

\begin{defi}\label{defi:representation}
  A {\bf representation of a hom-Lie algebra} $(\frkg,\br,_\frkg,\phi_\frkg)$ on
  a vector space $V$ with respect to $A\in\gl(V)$ is a linear map
  $\rho:\frkg\longrightarrow \gl(V)$, such that for any
  $x,y\in\frkg$, the following equalities are satisfied:
  \begin{itemize}
\item[\rm(i)] $\rho(\phi_\frkg(x))\circ A=A\circ
\rho(x);$
\item[\rm(ii)] $
    \rho([x,y]_\frkg)\circ
    A=\rho(\phi_\frkg(x))\circ\rho(y)-\rho(\phi_\frkg(y))\circ\rho(x).$
   \end{itemize}
\end{defi}

We denote a representation by $(V,A,\rho)$.  It is straightforward to see that $(\g,\phi_\g,\ad)$ is a representation,  called the {\bf adjoint representation}.

Given a representation $(V,A,\rho)$, we can construct a new hom-Lie algebra $\g\ltimes_{\rho} V=(\g\oplus V,[\cdot,\cdot]_{\g\ltimes_{\rho} V},\phi_\g\oplus A)$ as the semidirect product:
\begin{eqnarray*}
  (\phi_\g\oplus A)(x,u)&=&(\phi_\g(x),A(u)),\\
 ~ [(x,u),(y,v)]_{\g\ltimes_{\rho} V}&=&([x,y]_\g,\rho(x)(v)-\rho(y)(u)).
\end{eqnarray*}

Given a representation $(V,A,\rho)$, define $\rho^*:\frak g\rightarrow
\frak g\frak l(V^*)$  by
\begin{equation}\label{eq:dual}
\langle \rho^*(x)(\alpha), v\rangle =-\langle \alpha, \rho (x)(v)\rangle,\;\;\quad \forall~ x\in \frak g, \alpha\in V^*, v\in V.
\end{equation}
As observed in \cite{BM}, $(V^*,A^*,\rho^*)$ is not a representation of $\g$ on $V^*$ with respect to $A^*$ in general. However, it is easy to obtain the following result by definition.

\begin{lem} \label{lem:dual} Let $(\frak g,[\cdot,\cdot],\phi_\frak g)$ be a hom-Lie algebra and $( V,A,\rho)$ be a representation.
Then $(V^*,A^*,\rho^*)$ is a representation if and only if the following equations hold:
\begin{itemize}
\item[\rm(i)] $A\circ \rho (\phi_\frak g (x))=\rho (x)\circ A$;
\item[\rm(ii)] $A\circ \rho ([x,y]_\g)=\rho (x)\circ \rho (\phi_\frak g(y))-\rho (y)\circ \rho (\phi_\frak g(x))$.
\end{itemize}
\end{lem}

A representation $( V,A,\rho)$ is called {\bf admissible} if $(V^*,A^*,\rho^*)$ is also a representation, i.e. conditions (i) and (ii) in the above lemma are satisfied.  When we focus on the adjoint representation, we get

\begin{cor}
Let $(\frkg,\br__\frkg,\phi_\frkg)$ be a hom-Lie algebra. The adjoint representation $(\g,\phi_\g,\ad)$ is admissible if and only if the following two equations hold:
\begin{eqnarray}
  ~\label{eqn:coadjointrepcon1}[({\Id}-\phi_\frkg^2) (x),\phi_\frkg(y)]_\frkg&=&0,\\
~\label{eqn:coadjointrepcon2}[({\Id}-\phi_\frkg^2)(x),[\phi_\frkg(y),z]_\frkg]_\frkg&=&[({\Id}-\phi_\frkg^2)(y),[\phi_\frkg(x),z]_\frkg]_\frkg.
\end{eqnarray}
\end{cor}
\emptycomment{\pf $\ad^*:\frkg\longrightarrow\gl(\frkg^*)$ is a representation of
$\frkg$ on $\frkg^*$ with respect to $\phi_\frkg^*$ if and only if
\begin{eqnarray*}
  \phi_\frkg^*\circ\ad^*_u&=&\ad^*_{\phi_\g(x)}\circ\phi_\frkg^*,\quad\forall~u\in\frkg,\\
  \ad^*_{[x,y]_\g}\circ\phi_\frkg^*&=&\ad^*_{\phi_\g(x)}\circ\ad^*_v-\ad^*_{\phi_\frkg(v)}\circ\ad^*_u.
\end{eqnarray*}
By the first equality, we deduce that
$$
[u,\phi_\frkg(v)]_\frkg=\phi_\frkg([\phi_\g(x),v]_\frkg)=[\phi_\frkg^2(u),\phi_\frkg(v)]_\frkg,
$$
which is exactly the Eq. \eqref{eqn:coadjointrepcon1}.

By the second equality, we deduce that
$$
\phi_\frkg([[x,y]_\g,w]_\frkg)=[u,[\phi_\frkg(v),w]_\frkg]_\frkg-[v,[\phi_\g(x),w]_\frkg]_\frkg.
$$
By the hom-Jacobi identity, we have
$$
\phi_\frkg([[x,y]_\g,w]_\frkg)=[[\phi_\g(x),\phi_\frkg(v)]_\frkg,\phi_\frkg(w)]_\frkg
=[\phi_\frkg^2(u),[\phi_\frkg(v),w]_\frkg]_\frkg-[\phi_\frkg^2(v),[\phi_\g(x),w]_\frkg]_\frkg,
$$
which implies that Eq. \eqref{eqn:coadjointrepcon2} holds. \qed}

\begin{rmk}
  If  $(\frkg,\br__\frkg,\phi_\frkg)$ is a regular hom-Lie algebra,
  Eq. \eqref{eqn:coadjointrepcon1} implies that the image of
  $\Id-\phi_\frkg^2$, denoted by $\Img(\Id-\phi_\frkg^2)$, belongs to the center of $\frkg$, denoted by $\Cen(\frkg)$. If so, Eq.
  \eqref{eqn:coadjointrepcon2} holds naturally. Thus, for a regular
  hom-Lie algebra, the induced map $\ad^*$ is a representation if and only if
  $\Img(\Id-\phi_\frkg^2)\subset\Cen(\frkg)$.
\end{rmk}

\begin{defi}
A hom-Lie algebra  $(\frkg,\br__\frkg,\phi_\frkg)$ is called
{\bf admissible} if its adjoint representation is admissible, i.e. Eqs.
\eqref{eqn:coadjointrepcon1} and \eqref{eqn:coadjointrepcon2} are
satisfied.
\end{defi}

\begin{cor}\label{cor:ad2}
 Let  $(\frkg,\br__\frkg,\phi_\frkg)$ be a  regular admissible hom-Lie algebra,
then we have
 $$
\ad^*_x{\phi_\frkg^*}^2(\xi)=\ad^*_x\xi.
 $$
\end{cor}
\pf If $\frkg$ is regular, by Eq. \eqref{eqn:coadjointrepcon1}, we have
$[\phi_\frkg^2(x),y]_\frkg=[x,y]_\frkg$. Thus we have
\begin{eqnarray*}
\langle\ad^*_x{\phi_\frkg^*}^2(\xi),y\rangle&=&-\langle\xi,\phi_\frkg^2[x,y]_\frkg\rangle=-\langle\xi,[\phi_\frkg^2(x),\phi_\frkg^2(y)]_\frkg\rangle\\
&=&-\langle\xi,[x,\phi_\frkg^2(y)]_\frkg\rangle=-\langle\xi,[x,y]_\frkg\rangle=\langle\ad^*_x\xi,y\rangle,
\end{eqnarray*}
which implies the conclusion.  \qed

\begin{lem}\label{lem:admissible}
  Let  $(\frkg,\br__\frkg,\phi_\frkg)$ and
  $(\frkg^*,\br__{\frkg^*},\phi_{\frkg^*})$ be two admissible
  hom-Lie algebras. Then we have
  \begin{eqnarray*}
    [\xi, \ad^*_{\phi_\frkg(y)}\eta]_{\frkg^*}= [(\phi_\frkg^*)^2(\xi), \ad^*_{\phi_\frkg(y)}\eta]_{\frkg^*}.
  \end{eqnarray*}
\end{lem}
\pf For any  $x\in\frkg$, since $(\frkg,\br__\frkg,\phi_\frkg)$ is
admissible, we have
\begin{eqnarray*}
  \langle[\xi, \ad^*_{\phi_\frkg(y)}\eta]_{\frkg^*},x\rangle&=&-\langle\add^*_\xi x,
  \ad^*_{\phi_\frkg(y)}\eta\rangle=\langle[\phi_\frkg(y),\add^*_\xi
  x]_\frkg,\eta\rangle\\
&=&\langle[\phi_\frkg(y),\phi_\frkg^2\add^*_\xi
  x]_\frkg,\eta\rangle=\langle[\xi,
  {\phi_\frkg^*}^2\ad^*_{\phi_\frkg(y)}\eta]_{\frkg^*},x\rangle,
\end{eqnarray*}
which implies that
$$
[\xi, \ad^*_{\phi_\frkg(y)}\eta]_{\frkg^*}=[\xi,
  {\phi_\frkg^*}^2\ad^*_{\phi_\frkg(y)}\eta]_{\frkg^*}.
$$
Since  $(\frkg^*,\br__{\frkg^*},\phi_{\frkg^*})$ is also an
admissible
  hom-Lie algebra, we get
  $$
[\xi,
  {\phi_\frkg^*}^2\ad^*_{\phi_\frkg(y)}\eta]_{\frkg^*}=[{\phi_\frkg^*}^2(\xi),
  {\phi_\frkg^*}^2\ad^*_{\phi_\frkg(y)}\eta]_{\frkg^*}=[{\phi_\frkg^*}^2(\xi),
  \ad^*_{\phi_\frkg(y)}\eta]_{\frkg^*},
  $$
which implies the conclusion. \qed\vspace{3mm}

The set of {\bf $k$-cochains} on $\frkg$ with values in $V$, which we
denote by $C^k(\frkg;V)$, is the set of  skew-symmetric $k$-linear
maps from $\frkg\times\cdots\times\frkg$ $(k$-times$)$ to $V$:
$$
C^k(\frkg;V)\triangleq\{f:\wedge^k\frkg\longrightarrow V ~\mbox{is a
$k$-linear map}\}.
$$

A {\bf $k$-hom-cochain} on $\frkg$ with values in $V$ is defined to be a
$k$-cochain $f\in C^k(\frkg;V)$ such that it is compatible with
$\phi_\frkg$ and $A$ in the sense that $A\circ f=f\circ \phi_\frkg$,
i.e.
$$
A(f(x_1,\cdots,x_k))=f(\phi_\frkg(x_1),\cdots,\phi_\frkg(x_k)).
$$
Denote by $ C^k_{\phi_\frkg,A}(\frkg;V)$ the set of
$k$-hom-cochains:
$$
C^k_{\phi_\frkg,A}(\frkg;V)\triangleq\{f\in C^k(\frkg;V)|~A\circ
f=f\circ \phi_\frkg\}.
$$

In \cite{shenghomLie}, the author defined the coboundary operator
$\dM_{\rho}:C^k_{\phi_\frkg,A}(\frkg;V)\longrightarrow
C^{k+1}_{\phi_\frkg,A}(\frkg;V)$ by setting
\begin{eqnarray}
\nonumber  \dM_{\rho} f(x_1,\cdots,x_{k+1})&=&\sum_{i=1}^{k+1}(-1)^{i+1}\rho(\phi_\frkg^{k}(x_i))(f(x_1,\cdots,\widehat{x_i},\cdots,x_{k+1}))\\
 \label{eqn:d} &&+\sum_{i<j}(-1)^{i+j}f([x_i,x_j]_\frkg,\phi_\frkg(x_1)\cdots,\widehat{x_i},\cdots,\widehat{x_j},\cdots,\phi_\frkg(x_{k+1})).
\end{eqnarray}
The equality $\dM_{\rho}^2=0$ was proved in \cite{shenghomLie}.
Thus, we can obtain the cohomology of hom-Lie algebras.

Associated to the trivial representation, the set of $k$-cochains on
$\frkg$, which we denote by $C^k(\frkg)$, is the set of
skew-symmetric $k$-linear maps from $\frkg\times\cdots\times\frkg$
to $\mathbb R$, i.e. $ C^k(\frkg)=\wedge^k\frkg^*.$ The set of
$k$-hom-cochains is given by
$$
 C^k_{\phi_\frkg}(\frkg)=\{f\in\wedge^k\frkg^*|f\circ\phi_\frkg=f\}.
$$
The corresponding coboundary operator $\dM_T:
C^k_{\phi_\frkg}(\frkg)\longrightarrow  C^{k+1}_{\phi_\frkg}(\frkg)$ is
given by
$$
\dM_T
f(x_1,\cdots,x_{k+1})=\sum_{i<j}(-1)^{i+j}f([x_i,x_j]_\frkg,\phi_\frkg(x_1),\cdots,\widehat{x_i},\cdots,\widehat{x_j},\cdots,\phi_\frkg(x_{k+1})).
$$

\section{Matched pairs, hom-Lie bialgebras, and Manin triples}

Let $(\frkg,\br__\frkg,\phi_\frkg)$ and
$(\frkg',\br__{\frkg'},\phi_{\frkg'})$ be two hom-Lie algebras. Let
$\rho:\frkg\longrightarrow \gl(\frkg')$ and
$\rho':\frkg'\longrightarrow \gl(\frkg)$ be two linear maps. On the
direct sum of the underlying vector spaces, $\frkg\oplus \frkg'$, define $\phi_d:\frkg\oplus
\frkg'\longrightarrow\frkg\oplus \frkg'$ by
\begin{equation}
  \phi_d(x,x')=(\phi_\frkg(x),\phi_{\frkg'}(x')),
\end{equation}
and define a skew-symmetric bilinear map
$[\cdot,\cdot]_d:\wedge^2(\frkg\oplus \frkg')\longrightarrow
\frkg\oplus \frkg'$ by
\begin{equation}\label{eqn:brdouble}
  [(x,x'),(y,y')]_d=\big([x,y]_\frkg-\rho'(y')(x)+\rho'(x')(y),[x',y']_{\frkg'}+\rho(x)(y')-\rho(y)(x')\big).
\end{equation}

\begin{thm}\label{thm:matchedpair}
  With the above notations, $(\frkg\oplus
  \frkg',[\cdot,\cdot]_d,\phi_d)$ is a hom-Lie algebra if and only if $\rho$ and
  $\rho'$ are representations and
  \begin{eqnarray}
    \nonumber\rho'(\phi_{\frkg'}(x'))[x,y]_\frkg&=&[\rho'(x')(x),\phi_\frkg(y)]_\frkg+[\phi_\frkg(x),\rho'(x')(y)]_\frkg\\
   \label{eqn:matchcon1} &&+\rho'(\rho(y)(x'))(\phi_\frkg(x))-\rho'(\rho(x)(x'))(\phi_\frkg(y)),\\
    \nonumber\rho(\phi_{\frkg}(x))[x',y']_{\frkg'}&=&[\rho(x)(x'),\phi_{\frkg'}(y')]_{\frkg'}+[\phi_{\frkg'}(x'),\rho(x)(y')]_{\frkg'}\\
    \label{eqn:matchcon2}&&+\rho(\rho'(y')(x))(\phi_{\frkg'}(x'))-\rho(\rho'(x')(x))(\phi_{\frkg'}(y')).
  \end{eqnarray}
\end{thm}

\pf If $(\frkg\oplus
  \frkg',[\cdot,\cdot]_d,\phi_d)$  is a hom-Lie algebra, by the fact that $\phi_d$ is an
  algebra homomorphism with respect to $[\cdot,\cdot]_d$, we have
  $$
\phi_d[x,x']_d=[\phi_\frkg(x),\phi_{\frkg'}(x')]_d,
  $$
which implies that
\begin{eqnarray}
  \label{eqn:rhorep1}\phi_{\frkg'}\rho(x)(x')&=&\rho(\phi_\frkg(x))(\phi_{\frkg'}(x')),\\
 \label{eqn:rho'rep1}\phi_{\frkg}\rho'(x')(x)&=&\rho'(\phi_{\frkg'}(x'))(\phi_{\frkg}(x)).
\end{eqnarray}

Computing the hom-Jacobi identity for $x,y\in \frkg$ and
$x'\in\frkg'$, we have
\begin{eqnarray*}
  &&[[x,y]_d,\phi_{\frkg'}(x')]_d+[[y,x']_d,\phi_{\frkg}(x)]_d+[[x',x]_d,\phi_\g(y)]_d\\
  &=&\rho([x,y]_\frkg)(\phi_{\frkg'}(x'))-\rho'(\phi_{\frkg'}(x'))([x,y]_\frkg)+[\rho(y)(x')-\rho'(x')(y),\phi_{\frkg}(x)]_d\\
  &&+[\rho'(x')(x)-\rho(x)(x'),\phi_\g(y)]_d\\
&=&\rho([x,y]_\frkg)(\phi_{\frkg'}(x'))-\rho'(\phi_{\frkg'}(x'))([x,y]_\frkg)+\rho'(\rho(y)(x'))(\phi_{\frkg}(x))
-\rho(\phi_\g(x))\rho(y)(x')\\&&-[\rho'(x')(y),\phi_\g(x)]_\frkg
+[\rho'(x')(x),\phi_\g(y)]_\frkg-\rho'(\rho(x)(x'))(\phi_\g(y))+\rho(\phi_\g(y))\rho(x)(x')\\
&=&0,
\end{eqnarray*}
which implies that
\begin{eqnarray}
  \label{eqn:rhorep2}\rho([x,y]_\g)\circ\phi_{\frkg'}&=&\rho(\phi_\g(x))\rho(y)-\rho(\phi_\g(y))\rho(x),\\
 \nonumber\rho'(\phi_{\g'}(x'))([x,y]_\g)&=&\rho'(\rho(y)(x'))(\phi_\g(x))
-[\rho'(x')(y),\phi_\g(x)]_\frkg
\\\label{eqn:matchcon3}&&+[\rho'(x')(x),\phi_\g(y)]_\frkg-\rho'(\rho(x)(x'))(\phi_\g(y)).
\end{eqnarray}

Similarly, computing the hom-Jacobi identity for $x\in \frkg$ and
$x',y'\in\frkg'$, we get
\begin{eqnarray}
  \label{eqn:rho'rep2}\rho'([x',y']_{\g'})\circ\phi_{\frkg}&=&\rho'(\phi_{\g'}(x'))\rho'(y')-\rho'(\phi_{\frkg'}(y'))\rho'(x'),\\
 \nonumber\rho(\phi_\g(x))([x',y']_{\g'})&=&\rho(\rho'(y')(x))(\phi_{\g'}(x'))
-[\rho(x)(y'),\phi_{\g'}(x')]_{\frkg'}
\\\label{eqn:matchcon4}&&+[\rho(x)(x'),\phi_{\frkg'}(y')]_{\frkg'}-\rho(\rho'(x')(x))(\phi_{\frkg'}(y')).
\end{eqnarray}
By Eqs. \eqref{eqn:rhorep1} and \eqref{eqn:rhorep2}, we deduce that $\rho
$ is a representation of the hom-Lie algebra
$(\frkg,\br__\frkg,\phi_\frkg)$ on $\frkg'$ with respect to
$\phi_{\frkg'}$. By Eqs. \eqref{eqn:rho'rep1} and \eqref{eqn:rho'rep2},
we deduce that $\rho' $ is a representation of the hom-Lie algebra
$(\frkg',\br__{\frkg'},\phi_{\frkg'})$ on $\frkg$ with respect to
$\phi_{\frkg}$. Furthermore, Eqs. \eqref{eqn:matchcon3} and
\eqref{eqn:matchcon4} are exactly Eqs. \eqref{eqn:matchcon1} and
\eqref{eqn:matchcon2} respectively. This finishes the proof. \qed

\begin{defi}
  A {\bf matched pair of hom-Lie algebras}, which we denote by $(\frkg,\frkg';\rho,\rho')$, consists of two
  hom-Lie algebras $(\frkg,\br__\frkg,\phi_\frkg)$ and
$(\frkg',\br__{\frkg'},\phi_{\frkg'})$, together with
representations  $\rho:\frkg\longrightarrow \gl(\frkg')$ and
$\rho':\frkg'\longrightarrow \gl(\frkg)$ with respect to
$\phi_{\frkg'}$ and $\phi_\frkg$ respectively, such that the
compatibility conditions \eqref{eqn:matchcon1} and
\eqref{eqn:matchcon2} are satisfied.
\end{defi}

In the following, we concentrate on the case that $\frkg'$ is
$\frkg^*$, the dual space of $\frkg$, and
$\phi_{\frkg'}=\phi_\frkg^*$, $\rho=\ad^*$, $\rho'=\add^*$, where $\add^*$ is the dual map of $\add$.

For a hom-Lie algebra $(\frkg,\br__\frkg,\phi_\frkg)$ (resp.
$(\frkg^*,\br__{\frkg^*},\phi_{\frkg^*})$), let
$\Delta^*:\frkg^*\longrightarrow \wedge^2\frkg^*$ (resp.
$\Delta:\frkg\longrightarrow\wedge^2\frkg$) be the dual map of
$\br__\frkg:\wedge^2\frkg\longrightarrow\frkg$ (resp.
$\br__{\frkg^*}:\wedge^2\frkg^*\longrightarrow\frkg^*$), i.e.
$$
\langle\Delta^*(\xi),x\wedge
y\rangle=\langle\xi,[x,y]_\g\rangle,\quad
\langle\Delta(x),\xi\wedge \eta\rangle=\langle
x,[\xi,\eta]_{\frkg^*}\rangle.
$$
\begin{defi}\label{defi:homLiebi}
  A pair of admissible hom-Lie algebras $(\frkg,\br__\frkg,\phi_\frkg)$ and
  $(\frkg^*,\br__{\frkg^*},\phi_{\frkg}^*)$ is called a {\bf hom-Lie
  bialgebra} if
  \begin{eqnarray}
    \label{eqn:homLiebicon1}\langle\Delta[x,y]_\g,\phi_\frkg^*(\xi)\wedge\eta\rangle&=&\langle\ad_{\phi_\g(x)}\Delta(y),\phi_\frkg^*(\xi)\wedge\eta\rangle-
    \langle\ad_{\phi_\frkg(y)}\Delta(x),\phi_\frkg^*(\xi)\wedge\eta\rangle,\\
    \label{eqn:homLiebicon2}\langle \Delta^*[\xi,\eta]_{\frkg^*},\phi_\g(x)\wedge y\rangle&=&\langle\add_{\phi_\frkg^*(\xi)}\Delta^*(\eta),\phi_\g(x)\wedge y\rangle-
    \langle\add_{\phi_\frkg^*(\eta)}\Delta^*(\xi),\phi_\g(x)\wedge y\rangle.
  \end{eqnarray}
  Usually we denote a hom-Lie bialgebra simply by $(\frkg,\frkg^*)$.
\end{defi}

\begin{rmk}
  The notion of a hom-Lie bialgebra has already appeared in some
  references, e.g. \cite{Yao1}, where $\Delta$ is a 1-cocycle, i.e.
\begin{equation}\label{eqn:homLiebioldcon}
\Delta[x,y]_\g=\ad_{\phi_\g(x)}\Delta(y)-\ad_{\phi_\frkg(y)}\Delta(x).
\end{equation}

If $\phi_\frkg$ is not invertible, this condition is not the same as
the condition in Definition \ref{defi:homLiebi}. For example, if
$\phi_\frkg=0$, then for any skew-symmetric bilinear map
$[\cdot,\cdot]_\frkg$, $(\frkg,[\cdot,\cdot]_\frkg,0)$ is a hom-Lie
algebra. Any pair of hom-Lie algebras
$(\frkg,[\cdot,\cdot]_\frkg,0)$ and
$(\frkg^*,[\cdot,\cdot]_{\frkg^*},0)$ is a hom-Lie bialgebra (no
restriction on $[\cdot,\cdot]_{\frkg^*}$) in the sense of Definition
\ref{defi:homLiebi}. However, if $\phi_\frkg=0$, Eq.
\eqref{eqn:homLiebioldcon} implies that
$\Delta|_{[\frkg,\frkg]_\frkg}=0$, i.e. the cobracket
$[\cdot,\cdot]_{\frkg^*}$ must satisfy
$[\xi,\eta]_{\frkg^*}([x,y]_\g)=0$.
\end{rmk}

\begin{thm}\label{thm:biandmatchedpair}
  A pair of admissible hom-Lie algebras $(\frkg,\br__\frkg,\phi_\frkg)$ and
  $(\frkg^*,\br__{\frkg}^*,\phi_{\frkg}^*)$ is a hom-Lie
  bialgebra if and only if  $(\frkg,\br__\frkg,\phi_\frkg)$ and
  $(\frkg^*,\br__{\frkg^*},\phi_{\frkg}^*)$  is a matched pair of hom-Lie algebras, i.e. $(\frkg\oplus
  \frkg^*,[\cdot,\cdot]_d,\phi_\frkg\oplus\phi_\frkg^*)$ is a
   hom-Lie algebra, where $[\cdot,\cdot]_d$ is given by
   Eq. \eqref{eqn:brdouble}, in which $\rho=\ad^*$ and $\rho'=\add^*$.
\end{thm}
\pf By Theorem \ref{thm:matchedpair}, two admissible hom-Lie
algebras $(\frkg,\br__\frkg,\phi_\frkg)$ and
  $(\frkg^*,\br__{\frkg^*},\phi_{\frkg}^*)$  is a matched pair of hom-Lie
  algebras if and only if
\begin{eqnarray}
  \label{eqn:temp1} \add^*_{\phi_{\frkg}^*(\xi)}[x,y]_\g&=&[ \add^*_\xi x,\phi_\frkg(y)]_\frkg+[\phi_\g(x), \add^*_\xi y]_\frkg+ \add^*_{\ad _y^*\xi}\phi_\g(x)
   - \add^*_{\ad^* _x\xi}\phi_\frkg(y),\\
   \label{eqn:temp2} \ad^*_{\phi_\g(x)}[\xi,\eta]_{\frkg^*}&=&[\ad^*_x\xi,\phi_{\frkg}^*(\eta)]_{\frkg^*}+[\phi_{\frkg}^*(\xi),\ad^*_x \eta]_{\frkg^*}
   +\ad^*_{ \add^*_\eta x}\phi_{\frkg}^*(\xi)-\ad^*_{\add^*_\xi
x}\phi_{\frkg}^*(\eta).
  \end{eqnarray}
By Eq. \eqref{eqn:temp1}, we get
\begin{eqnarray*}
   0&=&\langle-\add^*_{\phi_{\frkg}^*(\xi)}[x,y]_\g+[ \add^*_\xi x,\phi_\frkg(y)]_\frkg+[\phi_\g(x), \add^*_\xi y]_\frkg+ \add^*_{\ad _y^*\xi}\phi_\g(x)
   - \add^*_{\ad^* _x\xi}\phi_\frkg(y),\eta\rangle\\
   &=&\langle[x,y]_\g,[\phi_{\frkg}^*(\xi),\eta]_{\frkg^*}\rangle-\langle\ad_{\phi_\frkg(y)}\add^*_\xi
   x,\eta\rangle+\langle\ad_{\phi_\g(x)}\add^*_\xi
   y,\eta\rangle\\&&-\langle \phi_\g(x),[\ad _y^*\xi,\eta]_{\frkg^*}
   \rangle+\langle \phi_\frkg(y),[\ad _x^*\xi,\eta]_{\frkg^*}
   \rangle\\
    &=&\langle[x,y]_\g,[\phi_{\frkg}^*(\xi),\eta]_{\frkg^*}\rangle-\langle
    x,[\xi,\ad_{\phi_\frkg(y)}^*\eta]_{\frkg^*}\rangle+\langle
    y,[\xi,\ad_{\phi_\g(x)}^*\eta]_{\frkg^*}\rangle\\
    &&-\langle x,[\phi_\frkg^*\ad _y^*\xi,\phi_\frkg^*\eta]_{\frkg^*}
   \rangle+\langle y,[\phi_\frkg^*\ad _x^*\xi,\phi_\frkg^*\eta]_{\frkg^*}
   \rangle\\
    &=&\langle[x,y]_\g,[\phi_{\frkg}^*(\xi),\eta]_{\frkg^*}\rangle-\langle
    x,[(\phi_\frkg^*)^2(\xi),\ad_{\phi_\frkg(y)}^*\eta]_{\frkg^*}\rangle+\langle
    y,[(\phi_\frkg^*)^2(\xi),\ad_{\phi_\g(x)}^*\eta]_{\frkg^*}\rangle\quad by~ Lemma ~\ref{lem:admissible}\\
    &&-\langle x,[\ad _{\phi_\frkg(y)}^*\phi_\frkg^*(\xi),\phi_\frkg^*\eta]_{\frkg^*}
   \rangle+\langle y,[\ad _{\phi_\g(x)}^*\phi_\frkg^*(\xi),\phi_\frkg^*\eta]_{\frkg^*}
   \rangle\\
    &=&\langle\Delta[x,y]_\g,\phi_{\frkg}^*(\xi)\wedge\eta\rangle-\langle
    \Delta x,(\phi_\frkg^*)^2(\xi)\wedge\ad_{\phi_\frkg(y)}^*\eta\rangle+\langle
     \Delta y,(\phi_\frkg^*)^2(\xi)\wedge\ad_{\phi_\g(x)}^*\eta\rangle\\
    &&-\langle \Delta  x,\ad _{\phi_\frkg(y)}^*\phi_\frkg^*(\xi)\wedge\phi_\frkg^*\eta
   \rangle+\langle  \Delta  y,\ad _{\phi_\g(x)}^*\phi_\frkg^*(\xi)\wedge\phi_\frkg^*\eta
   \rangle\\
   &=&\langle\Delta[x,y]_\g,\phi_{\frkg}^*(\xi)\wedge\eta\rangle-\langle
    \Delta x,\ad_{\phi_\frkg(y)}^*(\phi_\frkg^*(\xi)\wedge\eta)\rangle+\langle
     \Delta
     y,\ad_{\phi_\g(x)}^*(\phi_\frkg^*(\xi)\wedge\eta)\rangle\\
     &=&\langle\Delta[x,y]_\g,\phi_{\frkg}^*(\xi)\wedge\eta\rangle+\langle
   \ad_{\phi_\frkg(y)} \Delta x,\phi_\frkg^*(\xi)\wedge\eta\rangle-\langle
     \ad_{\phi_\g(x)}\Delta y,\phi_\frkg^*(\xi)\wedge\eta\rangle,
  \end{eqnarray*}
which is exactly Eq. \eqref{eqn:homLiebicon1}. Similarly, we could
deduce that  Eq. \eqref{eqn:temp2} is equivalent to Eq.
\eqref{eqn:homLiebicon2}. This finishes the proof. \qed

\begin{defi}
  A {\bf Manin triple of hom-Lie algebras} is a triple of hom-Lie algebras
  $(\frkk;\frkg,\frkg')$ together with a nondegenerate symmetric
  bilinear form $S(\cdot,\cdot)$ on $\frkk$ such that
  \begin{itemize}
    \item[$\bullet$] $S(\cdot,\cdot)$ is invariant, i.e. for any
    $x,y,z\in\frkk$, we have
    \begin{eqnarray}
      \label{eqn:invariant1}S([x,y]_\frkk,z)&=&S(x,[y,z]_\frkk);\\
     \label{eqn:invariant2} S(\phi_\frkk(x),y)&=&S(x,\phi_\frkk(y)).
    \end{eqnarray}

    \item[$\bullet$] $\frkg$ and $\frkg'$ are isotropic hom-Lie sub-algebras of
    $\frkk$, such that $\frkk=\frkg\oplus\frkg'$ as vector
    spaces.
  \end{itemize}
\end{defi}

Actually, $(\frkk,S)$ is a quadratic hom-Lie algebra, see \cite{BM} for more details. A {\bf Lagrangian hom-subalgebra} of  a quadratic hom-Lie algebra is defined to be a maximal isotropic hom-Lie subalgebra.

Let $(\frkg,\frkg^*)$ be a hom-Lie bialgebra, i.e. $\frkg$ and
$\frkg^*$ are admissible hom-Lie algebras such that Eqs.
\eqref{eqn:homLiebicon1} and \eqref{eqn:homLiebicon2} are satisfied.
By Theorem \ref{thm:biandmatchedpair}, we know that
$(\frkg\oplus\frkg^*,[\cdot,\cdot]_{\frkg\oplus\frkg^*},\phi_\frkg\oplus\phi_\frkg^*)$
is a hom-Lie algebra, where $[\cdot,\cdot]_{\frkg\oplus\frkg^*}$ is
given by
\begin{equation}\label{eqn:standbracket}
  [x+\xi,y+\eta]_{\frkg\oplus\frkg^*}=[x,y]_\frkg+[\xi,\eta]_{\frkg^*}+\ad^*_x\eta-\ad^*_y\xi+\add^*_\xi
  y-\add^*_\eta x.
\end{equation}
Furthermore, there is an obvious symmetric bilinear form on
$\frkg\oplus\frkg^*$:
\begin{equation}\label{eqn:standardform}
  S(x+\xi,y+\eta)=\langle x,\eta\rangle+\langle y,\xi\rangle.
\end{equation}
It is obvious that Eqs. \eqref{eqn:invariant1} and
\eqref{eqn:invariant2} are satisfied, i.e. the bilinear form defined
by Eq. \eqref{eqn:standardform} is invariant. Thus, we have
\begin{pro}
  Let $(\frkg,\frkg^*)$ be a hom-Lie bialgebra. Then
  $(\frkg\oplus\frkg^*;\frkg,\frkg^*)$ is a Manin triple of hom-Lie
  algebras.\label{pro:h-m}
\end{pro}

Conversely, if  $(\frkg\oplus\frkg^*;\frkg,\frkg^*)$ is a Manin triple of hom-Lie
  algebras with the invariant bilinear form $S$ given by Eq. \eqref{eqn:standardform}, then for any $x,y\in\frkg$ and $\xi,\eta\in\frkg^*$, due to the invariance of $S$,
we have
$$
S(\phi_\frkg(x),\xi)=S(x,\phi_{\frkg^*}(\xi))=\langle \phi_\frkg(x),\xi \rangle=\langle x,\phi_{\frkg^*}(\xi)\rangle,
$$
which implies that $\phi_{\frkg^*}=\phi_\frkg^*$, and
\begin{eqnarray*}
S([x,\xi]_{\frkg\oplus\frkg^*},y)&=&S([y,x]_\frkg,\xi)=\langle
-\ad_xy,\xi\rangle=\langle\ad_x^*\xi,y\rangle,\\
S([x,\xi]_{\frkg\oplus\frkg^*},\eta)&=&S(x,[\xi,\eta]_{\frkg^*})=\langle
\add_{\xi}\eta,x\rangle=\langle-\add^*_{\xi}x,\eta\rangle,
\end{eqnarray*}
which implies that
$$
[x,\xi]_{\frkg\oplus\frkg^*}=\ad_x^*\xi-\add^*_{\xi}x,
$$
that is, the hom-Lie bracket on $\frkg\oplus \frkg^*$ is given by Eq. \eqref{eqn:standbracket}.
Therefore, $(\g,\g^*;\ad^*,\add^*)$ is a matched pair of hom-Lie algebras and hence $(\frkg,\frkg^*)$ is a hom-Lie bialgebra. Note that
we deduce naturally that both $\frkg$ and $\frkg^*$ are admissible hom-Lie algebras.

Summarizing the above study, Theorem \ref{thm:biandmatchedpair} and Proposition \ref{pro:h-m}, we have the following
conclusion.

\begin{thm}
Let $(\frkg,\br__\frkg,\phi_\frkg)$ and
  $(\frkg^*,\br__{\frkg^*},\phi_{\frkg}^*)$ be two admissible hom-Lie algebras. Then the following conditions are equivalent.
\begin{itemize}
\item[\rm(i)]   $(\frkg, \frkg^*)$ is a hom-Lie bialgebra.
\item[\rm(ii)]   $(\g,\g^*;\ad^*,\add^*)$ is a matched pair of hom-Lie algebras.
\item[\rm(iii)]   $(\frkg\oplus\frkg^*;\frkg,\frkg^*)$ is a Manin triple of hom-Lie
  algebras with the invariant bilinear form \eqref{eqn:standardform}.
\end{itemize}
  \end{thm}

The bilinear form $S(\cdot,\cdot)$ given by
  Eq. \eqref{eqn:standardform} is called {\bf the standard bilinear form} on $\frkg\oplus
  \frkg^*$, and the hom-Lie bracket given by
  Eq. \eqref{eqn:standbracket} is called {\bf the standard hom-Lie bracket} on $\frkg\oplus
  \frkg^*$. The Manin triple $(\frkg\oplus\frkg^*;\frkg,\frkg^*)$ is
  called {\bf the standard Manin triple}.

\begin{pro}
  Any Manin triple of hom-Lie algebras $(\frkk;\frkg,\frkg')$ is
  isomorphic to the standard Manin triple
  $(\frkg\oplus\frkg^*;\frkg,\frkg^*)$.
\end{pro}
\pf Since $\frkg$ and $ \frkg'$ are isotropic under the
nondegenerate bilinear form $S(\cdot,\cdot)$ on
$\frkk=\frkg\oplus\frkg'$, we deduce that $\frkg'$ is isomorphic to
$\frkg^*$. Thus, $\frkk$ is isomorphic to $\frkg\oplus \frkg^*$ as
vector spaces. Transfer the nondegenerate bilinear form
$S(\cdot,\cdot)$ to $\frkg\oplus \frkg^*$, we obtain the standard
bilinear form  \eqref{eqn:standardform}.

Denote by $\phi_{\frkg^*}:\frkg^*\longrightarrow\frkg^*$ the induced
map from $\phi_{\frkg'}=\phi_\frkk|_{\frkg'}$. Due to the invariance of $S$,
we show that $\phi_{\frkg^*}=\phi_\frkg^*$.

It is straightforward to check that the hom-Lie bracket on
$\frkk$ can be transferred into $\frkg\oplus\frkg^*$, which is
the standard hom-Lie bracket
\eqref{eqn:standbracket}. Thus, $(\frkk;\frkg,\frkg')$ is
  isomorphic to the standard Manin triple
  $(\frkg\oplus\frkg^*;\frkg,\frkg^*)$.
 \qed

\section{$r$-matrices }

For any $r\in\wedge^2\frkg$, the induced skew-symmetric linear map
$r^\sharp:\frkg^*\longrightarrow\frkg$ is defined by
$$
\langle r^\sharp(\xi),\eta\rangle=\langle r,\xi\wedge\eta\rangle.
$$

\begin{defi}
  A {\bf coboundary hom-Lie bialgebra} is a hom-Lie bialgebra
  $(\frkg,\frkg^*)$ such that \begin{equation}\label{eqn:deltar}
   \langle
   \Delta(x),\phi_\frkg^*(\xi)\wedge\eta\rangle=\langle[x,r]_\frkg,\phi_\frkg^*(\xi)\wedge\eta\rangle,
  \end{equation}
  where $r\in\wedge^2\frkg$ satisfies
  \begin{equation}\label{eqn:r}
    \phi_\frkg\circ r^\sharp\circ \phi_\g^*=r^\sharp.
  \end{equation}

  A {\bf triangular hom-Lie bialgebra} is a coboundary hom-Lie bialgebra,
  in which $r$ is a solution of the following {\bf classical
  hom-Yang-Baxter equation}
  \begin{equation}\label{eqn:rmatrixcon}
    [r,r]_\g=0.
  \end{equation}
\end{defi}

\begin{rmk}It is easy to see that Eq. \eqref{eqn:r} is equivalent to that
$ \phi_\g^{\otimes2}r=r, $ which means that $r$ is a $0$-cochain,
see \cite{shenghomLie} for more details. It is also easy to show that if $\phi_\g$ is orthogonal as a linear transformation on $\frak g$, then Eq. \eqref{eqn:r} holds if and only if $\phi_\frkg\circ r^\sharp =r^\sharp\circ\phi_\g^*$.
\end{rmk}

\begin{pro}\label{pro:brdualr}
  With the above notations, let $(\frkg,\frkg^*)$  be a coboundary
  hom-Lie bialgebra. Then for any $\xi\in\Img(\phi_\frkg^*)$ and $\eta\in\frkg^*$, we have
  \begin{equation}\label{eqn:brdualr}
    ~[\xi,\eta]_{\frkg^*}=\ad^*_{r^\sharp\circ\phi_\frkg^*(\xi)}\eta-\ad^*_{r^\sharp\circ\phi_\frkg^*(\eta)}\xi.
  \end{equation}
\end{pro}
\pf To be simple, assume that $r=X\wedge Y$ is a monomial. We have
\begin{eqnarray*}
  &&\langle x,[\xi,\eta]_{\frkg^*}\rangle\\&=&\langle\Delta(x),\xi\wedge
  \eta\rangle=\langle[x,r]_\g,\xi\wedge
  \eta\rangle\\
  &=&\langle[x,X]_\g\wedge\phi_\frkg(Y)+\phi_\frkg(X)\wedge[x,Y]_\g,\xi\wedge\eta\rangle\\
  &=&\langle[x,X]_\g,\xi\rangle\langle\phi_\frkg(Y),\eta\rangle-\langle[x,X]_\g,\eta\rangle\langle\phi_\frkg(Y),\xi\rangle+\langle[x,Y]_\g,
  \eta\rangle\langle\phi_\frkg(X),\xi\rangle-\langle[x,Y]_\g,\xi\rangle\langle\phi_\frkg(X),\eta\rangle\\
  &=&-\langle[x,\langle\phi_\frkg(Y),\xi\rangle X]_\g,\eta\rangle+\langle[x,\langle\phi_\frkg(X),\xi\rangle
  Y]_\g,\eta\rangle+\langle[x,\langle \phi_\frkg(Y),\eta\rangle X]_\g,\xi\rangle-\langle[x,\langle \phi_\frkg(X),\eta\rangle
  Y]_\g,\xi\rangle\\
  &=&-\langle[x,\langle Y,\phi_\frkg^*(\xi)\rangle X]_\g,\eta\rangle+\langle[x,\langle X,\phi_\frkg^*(\xi)\rangle
  Y]_\g,\eta\rangle+\langle[x,\langle Y,\phi_\frkg^*(\eta)\rangle X]_\g,\xi\rangle-\langle[x,\langle  X,\phi_\frkg^*(\eta)\rangle
  Y]_\g,\xi\rangle\\
  &=&\langle [x,r^\sharp\circ\phi_\frkg^*(\xi)]_\g,\eta\rangle-\langle
  [x,r^\sharp\circ\phi_\frkg^*(\eta)]_\g,\xi\rangle\\
  &=&\langle
  x,\ad^*_{r^\sharp\circ\phi_\frkg^*(\xi)}\eta-\ad^*_{r^\sharp\circ\phi_\frkg^*(\eta)}\xi\rangle,
\end{eqnarray*}
which implies the conclusion. \qed\vspace{3mm}

Even though formula \eqref{eqn:brdualr} does not hold for any
$\xi,\eta\in\g^*$, but we have

\begin{cor}\label{cor:bracket}
 With the above notations, let $(\frkg,\frkg^*)$  be a coboundary
  hom-Lie bialgebra. Then for any $\xi,\eta\in\frkg^*$, we have
  \begin{equation}\label{eqn:brdualr1}
    ~\phi^*_\frkg[\xi,\eta]_{\frkg^*}=\phi^*_\frkg\big(\ad^*_{r^\sharp\circ\phi_\frkg^*(\xi)}\eta-\ad^*_{r^\sharp\circ\phi_\frkg^*(\eta)}\xi\big).
  \end{equation}
\end{cor}
\pf By Proposition \ref{pro:brdualr}, we have
\begin{eqnarray*}
\phi^*_\frkg[\xi,\eta]_{\frkg^*}=[\phi_\frkg^*(\xi),\phi_\frkg^*(\eta)]_{\frkg^*}=
 \ad^*_{r^\sharp\circ(\phi_\frkg^*)^2(\xi)}\phi_\frkg^*(\eta)-\ad^*_{r^\sharp\circ(\phi_\frkg^*)^2(\eta)}\phi_\frkg^*(\xi).
\end{eqnarray*}
In the following we prove that
\begin{equation}\label{eqn:temp3}
\ad^*_{r^\sharp\circ(\phi_\frkg^*)^2(\xi)}\phi_\frkg^*(\eta)=\phi_\frkg^*\ad^*_{
r^\sharp\circ\phi_\frkg^*(\xi)}\eta.
\end{equation}
For any $x\in \frkg$, we have
\begin{eqnarray*}
 \langle\ad^*_{r^\sharp\circ(\phi_\frkg^*)^2(\xi)}\phi_\frkg^*(\eta),x\rangle&=&-\langle\phi_\frkg^*(\eta),[r^\sharp\circ(\phi_\frkg^*)^2(\xi),x]_\g\rangle\\
 &=&-\langle\eta,[\phi_\g\circ
 r^\sharp\circ(\phi_\frkg^*)^2(\xi),\phi_\g(x)]_\g\rangle\\
 &=&-\langle\eta,[
 r^\sharp\circ\phi_\frkg^*(\xi),\phi_\g(x)]_\g\rangle\qquad\qquad\qquad\mbox{by ~Eq.~\eqref{eqn:r}, }\\
 &=&-\langle\eta,[
 \phi_\g^2\circ
 r^\sharp\circ\phi_\frkg^*(\xi),\phi_\g(x)]_\g\rangle\qquad\qquad\mbox{by ~Eq.~\eqref{eqn:coadjointrepcon1}, }\\
&=&-\langle\phi_\g^*(\eta),[
 \phi_\g\circ
 r^\sharp\circ\phi_\frkg^*(\xi),x]_\g\rangle\\
 &=&\langle\ad^*_{\phi_\g\circ
 r^\sharp\circ\phi_\frkg^*(\xi)}\phi_\g^*(\eta),x\rangle\\
 &=&\langle\phi^*_\g\ad^*_{
 r^\sharp\circ\phi_\frkg^*(\xi)}\eta,x\rangle.
\end{eqnarray*}
The last equality holds since $\g$ is admissible, i.e. $\ad^*$ is a
representation. Then we obtain Eq. \eqref{eqn:temp3}. The proof is
finished. \qed

\begin{cor}\label{cor:difference}
 Let $(\frkg,\frkg^*)$  be a coboundary
  hom-Lie bialgebra. If $\g$ is regular, then we have
  \begin{equation}
    \ad_{r^\sharp\circ\phi_\g^*(\xi)}=\ad_{\phi_\g\circ
    r^\sharp(\xi)},
  \end{equation}
  i.e. the image of $r^\sharp\circ\phi_\g^*-\phi_\g\circ
    r^\sharp$ belongs to the center of $\g$.
\end{cor}
\pf We have proved that
$$
\ad^*_{r^\sharp\circ(\phi_\frkg^*)^2(\xi)}\phi_\frkg^*(\eta)=\ad^*_{\phi_\g\circ
 r^\sharp\circ\phi_\frkg^*(\xi)}\phi_\g^*(\eta).
$$
Thus, if $\g$ is regular, we have
$$
\ad^*_{r^\sharp\circ\phi_\g^*(\xi)}=\ad^*_{\phi_\g\circ
    r^\sharp(\xi)},
$$
which implies the conclusion. \qed

\begin{lem}\label{lem:wellknownformula}
  Let $\g$  be a regular
  admissible hom-Lie algebra. If $r\in\wedge^2\g$ satisfies Eq. \eqref{eqn:r} and
  $[\cdot,\cdot]_\frkg^*:\wedge^2\g^*\longrightarrow\g^*$ is given
  by Eq. \eqref{eqn:brdualr}, then we have
  \begin{equation}
    [r^\sharp\circ\phi^*_\g(\xi),r^\sharp\circ\phi^*_\g(\eta)]_\g-r^\sharp\circ\phi^*_\g[\xi,\eta]_{\g^*}=\half[r,r]_\g(\xi,\eta).
  \end{equation}
\end{lem}
\pf For any $\gamma\in\g^*$, by Corollary \ref{cor:difference}, we
have
\begin{eqnarray*}
  \langle r^\sharp\circ\phi^*_\g[\xi,\eta]_{\g^*},\gamma\rangle&=&\langle\ad^*_{r^\sharp\circ\phi_\g^*(\eta)}\xi-\ad^*_{r^\sharp\circ\phi_\g^*(\xi)}\eta,\phi_\g\circ
  r^\sharp(\gamma)\rangle\\
  &=&\langle\eta,[r^\sharp\circ\phi_\g^*(\xi),\phi_\g\circ
  r^\sharp(\gamma)]_\g\rangle-\langle\xi,[r^\sharp\circ\phi_\g^*(\eta),\phi_\g\circ
  r^\sharp(\gamma)]_\g\rangle\\
  &=&\langle\eta,[r^\sharp\circ\phi_\g^*(\xi),r^\sharp\circ\phi_\g^*(\gamma)]_\g\rangle-\langle\xi,[r^\sharp\circ\phi_\g^*(\eta),r^\sharp\circ\phi_\g^*(\gamma)]_\g\rangle.
\end{eqnarray*}
Then it is straightforward to deduce that
$$
\half[r,r]_\g(\xi,\eta,\gamma)=-\langle\eta,[r^\sharp\circ\phi_\g^*(\xi),r^\sharp\circ\phi_\g^*(\gamma)]_\g\rangle
+\langle\xi,[r^\sharp\circ\phi_\g^*(\eta),r^\sharp\circ\phi_\g^*(\gamma)]_\g\rangle
+\langle\gamma,[r^\sharp\circ\phi^*_\g(\xi),r^\sharp\circ\phi^*_\g(\eta)]_\g\rangle.
$$
We leave it to readers. The proof is completed. \qed\vspace{3mm}

For a coboundary
  hom-Lie bialgebra $(\frkg,\frkg^*)$, if
$\g$ is regular, Eq. \eqref{eqn:deltar} reduces to
\begin{equation}\label{eqn:deltarnew}
  \Delta(x)=[x,r]_\g,
\end{equation}
and for any $\xi,\eta\in\frkg^*$, $[\xi,\eta]_{\g^*}$ is given by
Eq. \eqref{eqn:brdualr}.

\begin{thm}
   Let $\g$ be a regular
  admissible hom-Lie algebra. Define a skew-symmetric bilinear map $[\cdot,\cdot]_{\g^*}\wedge^2\g^*\longrightarrow\g^*$ by Eq. \eqref{eqn:brdualr},
   for some $r\in\wedge^2\g$ satisfying
  Eq. \eqref{eqn:r}. Then $(\g^*,[\cdot,\cdot]_{\g^*},\phi_{\g*})$ is a hom-Lie algebra if and only if
  \begin{equation}\label{eqn:conditionr}
    \ad_x[r,r]_\g=0.
  \end{equation}
  Under this condition, $(\frkg,\g^*)$ is a coboundary hom-Lie
  bialgebra.
\end{thm}
\pf The bracket $[\cdot,\cdot]_{\g^*}$ is determined by the
cobracket $\Delta(x)=[x,r]_\g$. The condition $\phi_\g^*$ is an
algebra homomorphism with respect to $[\cdot,\cdot]_{\g^*}$ is
equivalent to $\phi_\g^{\otimes2}\Delta(x)=\Delta(\phi_\g(x))$. If
$r$ satisfies Eq. \eqref{eqn:r}, equivalently, $\phi_\g^{\otimes2}r=r$,
then we have
$$\Delta(\phi_\g(x))=[\phi_\g(x),r]_\g=[\phi_\g(x),\phi_\g^{\otimes2}r]_\g=\phi_\g^{\otimes2}[x,r]_\g=\phi_\g^{\otimes2}\Delta(x).$$

If $[\cdot,\cdot]_{\g^*}$ satisfies the hom-Jacobi identity, we have
\begin{eqnarray*}
 && J(\xi,\eta,\gamma)=[\phi_\g^*(\xi),[\eta,\gamma]_{\g^*}]_{\g^*}+
  [\phi_\g^*(\eta),[\gamma,\xi]_{\g^*}]_{\g^*}+
  [\phi_\g^*(\gamma),[\xi,\eta]_{\g^*}]_{\g^*}\\
  &=&\ad^*_{r^\sharp\circ(\phi_\g^*)^2(\xi)}(\ad^*_{r^\sharp\circ\phi^*_\g(\eta)}\gamma-\ad^*_{r^\sharp\circ\phi^*_\g(\gamma)}\eta)-
  \ad^*_{r^\sharp\circ\phi^*_\g[\eta,\gamma]_{\g^*}}\phi^*_\g(\xi)+c.p.\\
  &=&\ad^*_{\phi_\g\circ r^\sharp\circ\phi_\g^*(\xi)}(\ad^*_{r^\sharp\circ\phi^*_\g(\eta)}\gamma-\ad^*_{r^\sharp\circ\phi^*_\g(\gamma)}\eta)-
  \ad^*_{r^\sharp\circ\phi^*_\g[\eta,\gamma]_{\g^*}}\phi^*_\g(\xi)+c.p.\quad\mbox{by~Corollary
  ~\ref{cor:difference}}\\
  &=&\ad^*_{[r^\sharp\circ\phi_\g^*(\xi),r^\sharp\circ\phi_\g^*(\eta)]_\g}\phi^*_\g(\gamma)-\ad^*_{r^\sharp\circ\phi^*_\g[\xi,\eta]_{\g^*}}\phi^*_\g(\gamma)+c.p.\\
  &=&\half
  \ad^*_{[r,r]_\g(\xi,\eta)}\phi^*_\g(\gamma)+c.p.\qquad\mbox{by~Lemma~\ref{lem:wellknownformula}}.
\end{eqnarray*}
Thus, we have
\begin{eqnarray*}
  \langle
 2
 J(\xi,\eta,\gamma),x\rangle&=&-\langle\phi^*_\g(\gamma),[[r,r]_\frkg(\xi,\eta),x]_\frkg\rangle+c.p.=-\langle\ad_x^*\phi^*_\g(\gamma),[r,r]_\frkg(\xi,\eta)\rangle+c.p.\\
&=&-[r,r]_\frkg(\xi,\eta,\ad_x^*\phi^*_\g(\gamma))+c.p.\\
&=&-[r,r]_\frkg(\xi,\eta,\ad_x^*{\phi^*_\g}^{-1}(\gamma))+c.p.
\quad\mbox{by ~~Corollary ~\ref{cor:ad2}}\\
&=&\langle
-[r,r]_\frkg,\ad_x^*({\phi^*_\g}^{-1}(\xi)\wedge{\phi^*_\g}^{-1}(\eta)\wedge{\phi^*_\g}^{-1}(\gamma))\rangle\\
&=&\langle
\ad_x[r,r]_\frkg,{\phi^*_\g}^{-1}(\xi)\wedge{\phi^*_\g}^{-1}(\eta)\wedge{\phi^*_\g}^{-1}(\gamma)\rangle,
\end{eqnarray*}
which implies that under the condition \eqref{eqn:r},
$(\g,[\cdot,\cdot]_{\g^*},\phi^*_\g)$ is a hom-Lie algebra if and
only if Eq. \eqref{eqn:conditionr} is satisfied.

At last, if $\Delta(x)=[x,r]_\g$, then the compatibility conditions
in Definition \ref{defi:homLiebi} hold naturally. The proof is
completed. \qed

\begin{rmk}
  A similar result was also given in \cite{Yao1} (see Theorem 4.5 therein). Here we give a slightly different approach.
\end{rmk}

Let $\g$  be a regular
  admissible hom-Lie algebra and $r\in \wedge^2 \frak g$ be invertible (that is, $r^\sharp$ is invertible). Define $B\in\wedge^2\g^*$ by
  $$B(x,y)=\langle {r^\sharp}^{-1} (x), y\rangle.$$

\begin{pro}\label{pro:2-cocycle}
With the above notations,   $r$ satisfies the
classical hom-Yang-Baxter equation~\eqref{eqn:rmatrixcon} if and only if
\begin{equation}\label{eq:2-cocycle}
B(x, \phi_\frak g[y,z]_\g)+B(z, \phi_\frak g ([x,x]_\g))+B(y, \phi_\frak g ([z,x]_\g))=0.
\end{equation}
\end{pro}
\pf
For any $x,y,z\in \frak g$, set $x=r^\sharp (\xi), y=r^\sharp (\eta), z=r^\sharp (\gamma)$. By Corollary \ref{cor:difference} and Lemma \ref{lem:wellknownformula}, if $r$ satisfies the classical hom-Yang-Baxter equation,  we have
\begin{eqnarray*}
B(x, \phi_\frak g[y,z]_\g)&=&\langle \xi, \phi_\frak g[r^\sharp (\eta),r^\sharp (\gamma)]_\g\rangle=\langle \xi, [r^\sharp\circ \phi_\frak g^*(\eta), r^\sharp\circ \phi_\frak
g^*(\gamma)]_\g\rangle\\
&=& \langle \xi, r^\sharp\circ \phi_\frak g^*({\rm ad}^*_{r^\sharp\circ \phi_\frak
g^*(\eta)}\gamma-{\rm ad}^*_{r^\sharp\circ \phi_\frak g^*(\gamma)}\eta)\rangle\\
&=& \langle -\phi_\frak g\circ r^\sharp (\xi), {\rm ad}^*_{r^\sharp\circ \phi_\frak
g^*(\eta)}\gamma-{\rm ad}^*_{r^\sharp\circ \phi_\frak g^*(\gamma)}\eta\rangle\\
&=& \langle \gamma, [\phi_\frak g\circ r^\sharp (\eta), \phi_\frak g\circ r^\sharp (\xi)]_\g\rangle
-\langle \eta, [\phi_\frak g\circ r^\sharp (\gamma), \phi_\frak g\circ r^\sharp (\xi)]_\g\rangle\\
&=& B(z, \phi_\frak g ([y,x]_\g))-B(y, \phi_\frak g ([z,x]_\g)).
\end{eqnarray*}
The converse part can be obtained similarly.
\qed

\begin{rmk} {\rm If $\phi_\g$ is orthogonal, or the center of $\g$ is zero, by Corollary \ref{cor:difference},  we have $r^\sharp\circ\phi_\g^*=\phi_\g\circ r^\sharp$. Thus, we have
\begin{eqnarray*}
B(\phi_\frak g (x), y)=\langle {r^\sharp}^{-1}\circ \phi_\frak g (x), y\rangle
=\langle \phi_\frak g^*\circ{r^\sharp}^{-1}(x), y\rangle
=\langle {r^\sharp}^{-1}(x), \phi_\frak g (y)\rangle
=B(x, \phi_\frak g (y)).
\end{eqnarray*}
Therefore Eq.~\eqref{eq:2-cocycle} can be rewritten as
\begin{equation}\label{eq:2-cocycle-new}
B(\phi_\frak g (x),[y,z]_\g)+B(\phi_\frak g (y), [z,x]_\g)+B(\phi_\frak g (z), [x,y]_\g)=0,
\end{equation}
which means that $B$ is  a {\bf 2-cocycle} on $\frak g$ (\cite{shenghomLie}).}
\end{rmk}

At the end of this section, we consider Lagrangian hom-subalgebras of a hom-Lie bialgebra.
\begin{thm}\label{thm:R}
Let $(\frkg,\frkg^*)$ be a regular hom-Lie bialgebra and
$R\in\wedge^2\frkg$, the graph of $ R^\sharp\circ\phi_\g^*$, which
we denote by $\huaG_R$, is a Lagrangian hom-subalgebra of $\frkg\oplus
\frkg^*$ if $R^\sharp\circ\phi_{\frkg}^*=\phi_\frkg\circ R^\sharp$
and the following Maurer-Cartan type equation is satisfied:
\begin{equation}\label{eqn:mc}
  d_T^*R+\half [R,R]_\g=0,
\end{equation}
where $d^*_T$ is the coboundary operator for the hom-Lie algebra
$\g^*$ associated to the trivial representation.
\end{thm}
\pf It is easy to see that
$R^\sharp\circ\phi_{\frkg^*}=\phi_\frkg\circ R^\sharp$ guarantees
that $(\phi_\g\oplus\phi_{\g^*})\huaG_R\subset \huaG_R$, and $\huaG_R$ is maximal isotropic. Next we
show that $\huaG_R$ is closed under the bracket operation
$[\cdot,\cdot]_{\g\oplus\g^*}$ given by Eq. \eqref{eqn:standbracket} if and only if
Eq. \eqref{eqn:mc} holds. First we have
\begin{eqnarray*}
  [R^\sharp\circ\phi_\g^*(\xi)+\xi,R^\sharp\circ\phi_\g^*(\eta)+\eta]_{\g\oplus\g^*}&=&[R^\sharp\circ\phi_\g^*(\xi),R^\sharp\circ\phi_\g^*(\eta)]_\g
  -\add^*_{\eta}R^\sharp\circ\phi_\g^*(\xi)+\add^*_{\xi}R^\sharp\circ\phi_\g^*(\eta)\\
  &&+[\xi,\eta]_{\g^*}+\ad^*_{R^\sharp\circ\phi_\g^*(\xi)}\eta-\ad^*_{R^\sharp\circ\phi_\g^*(\eta)}\xi.
\end{eqnarray*}
$\huaG_R$ is closed if and only if
$$
[R^\sharp\circ\phi_\g^*(\xi),R^\sharp\circ\phi_\g^*(\eta)]_\g
  -\add^*_{\eta}R^\sharp\circ\phi_\g^*(\xi)+\add^*_{\xi}R^\sharp\circ\phi_\g^*(\eta)
  =R^\sharp\circ\phi_\g^*\Big(
  [\xi,\eta]_{\g^*}+\ad^*_{R^\sharp\circ\phi_\g^*(\xi)}\eta-\ad^*_{R^\sharp\circ\phi_\g^*(\eta)}\xi\Big).
$$
Note that
\begin{eqnarray*}
&&[R^\sharp\circ\phi_\g^*(\xi),R^\sharp\circ\phi_\g^*(\eta)]_\g-R^\sharp\circ\phi_\g^*\big(
  \ad^*_{R^\sharp\circ\phi_\g^*(\xi)}\eta-\ad^*_{R^\sharp\circ\phi_\g^*(\eta)}\xi\big)\\&=&[R^\sharp\circ\phi_\g^*(\xi),R^\sharp\circ\phi_\g^*(\eta)]_\g
  -R^\sharp\circ\phi_\g^*([\xi,\eta]_R)\\
  &=&\half [R,R]_\g(\xi,\eta),
\end{eqnarray*}
and
\begin{eqnarray*}
 &&\langle -\add^*_{\eta}R^\sharp\circ\phi_\g^*(\xi)+\add^*_{\xi}R^\sharp\circ\phi_\g^*(\eta)-R^\sharp\circ\phi_\g^*
  [\xi,\eta]_{\g^*},\gamma\rangle\\
  &=&R(\phi^*_\g(\xi),[\eta,\gamma]_{\g^*})-R(\phi_\g^*(\eta),[\xi,\gamma]_{\g^*})-R([\xi,\eta]_{\g^*},\phi^*_\g(\gamma))\\
  &=&d^*_TR(\xi,\eta,\gamma).
\end{eqnarray*}
Thus, $\huaG_R$ is closed if and only if Eq. \eqref{eqn:mc} holds. \qed

\section{Interpretation of $r$-matrices in terms of operator forms}

In this section, we give a further interpretation of the classical
hom-Yang-Baxter equation.

\begin{defi}{\rm
Let $(\frak g,[\cdot,\cdot]_\g,\phi_\g)$ be a hom-Lie algebra and $( V,A,\rho)$ be a representation. A linear map $T:V\rightarrow \g$ is called an {\bf $\mathcal O$-operator} associated to $( V,A,\rho)$ if $T$ satisfies
\begin{itemize}
  \item[\rm(i)] $T\circ A=\phi_\g\circ T$,
  \item[\rm(ii)] $[T(u),T(v)]_\g=T(\rho (T(u))v-\rho(T(v)u)),\quad \forall~ u,v\in V.$
\end{itemize}}
\end{defi}

Recall that a notion of a hom-Nijenhuis operator for a hom-Lie algebra $(\g,[\cdot,\cdot]_\g,\phi_\g)$ was introduced in \cite{shenghomLie}, which could give a trivial deformation. See \cite{D} for more details about Nijenhuis operators for Lie algebras and their applications in nonlinear evolution equations. A {\bf hom-Nijenhuis operator} is a linear map $N:\g\longrightarrow\g$ satisfying $N\circ \phi_\g=\phi_\g\circ N$ and
$$
[N(x),N(y)]_\g=N\big([N(x),y]_\g+[x,N(y)]_\g-N[x,y]_\g\big).
$$
It is easy to show that
\begin{lem}
  A linear map $T:V\rightarrow \g$ is  an {$\mathcal O$-operator} associated to $( V,A,\rho)$ if and only if $\left(\begin{array}{cc}0&T\\
  0&0\end{array}\right)$ is a hom-Nijenhuis operator for the hom-Lie algebra $\g\ltimes_\rho V$.
\end{lem}

\begin{ex} {\rm Let $\frak g$ be a hom-Lie algebra. A multiplicative {\bf Rota-Baxter operator} (of weight zero) $R$ on $\frak g$ is an $\mathcal O$-operator associated to the adjoint representation,
that is, $R$ satisfies (\cite{MY})
$$[R(x),R(y)]_\g=R([R(x),y]_\g-[R(y),x]_\g),\quad \forall~ x,y\in \frak g.$$}
\end{ex}

\begin{ex} \label{ex:2-cocycle}{\rm Let $\frak g$ be a regular admissible hom-Lie algebra. Suppose that $r\in \wedge^2\frak g
$ satisfies Eq. \eqref{eqn:r}. By Lemma \ref{lem:wellknownformula}, $r$ satisfies the
classical hom-Yang-Baxter equation~\eqref{eqn:rmatrixcon} if and only if $r^\sharp \circ\phi_\frak g^*$
is an $\mathcal O$-operator associated to the co-adjoint representation ${\rm ad}^*$.}
\end{ex}

There is a class of Hom-algebra structures behind:

\begin{defi}{\rm (\cite{MS2,Yao4})}  A {\bf hom-left-symmetric algebra} is a triple $(V, \cdot, \psi)$ consisting of a linear space $V$,
a bilinear map $\cdot: V\otimes V\rightarrow V$ and an algebra homomorphism $\psi:V\rightarrow V$ satisfying
\begin{equation}
(u\cdot v)\cdot \psi(w)-\psi(u) \cdot (v\cdot w)=(v\cdot u)\cdot \psi(w)-\psi(v)\cdot (u\cdot w),~\quad\forall~ u,v,w\in V.
\end{equation}
It is called {\bf regular (involutive)}, if $\psi$ is nondegenerate (satisfies
$\psi^2=\Id$).
\end{defi}

The following conclusion is obvious.

\begin{pro}\label{prop:left-Lie} Let $(V,\cdot,\psi)$ be a hom-left-symmetric algebra.
\begin{enumerate}
\item[\rm(i)] $(\g(V),[\cdot,\cdot]_V,\psi)$ is a hom-Lie algebra, where $\g(V)=V$ as a vector space, and the bracket $[\cdot,\cdot]_V$ is given by
 $$[u,v]_V=u\cdot v-v\cdot u.$$
 We call $(\g(V),[\cdot,\cdot]_V,\psi)$ the {\bf commutator hom-Lie algebra}.
\item[\rm(ii)] Let $L:V\rightarrow \gl(V)$ be a linear map defined by $L(u)(v)=u\cdot v$ for any $u,v\in V$. Then
$(V,\psi,L)$ is a representation of the hom-Lie algebra $(\g(V),[\cdot,\cdot]_V,\psi)$ on $V$.
\end{enumerate}
\end{pro}

\begin{pro}\label{prop:lie-lsa}
Let $(\frak g,[\cdot,\cdot],\phi_\frak g)$ be a hom-Lie algebra and $ (V,A,\rho)$ be a representation of $\g$. Suppose that $T:V\rightarrow \g$ is  an  $\mathcal O$-operator.
Then there exists a hom-left-symmetric algebra structure $(V,\cdot, A)$ on $V$ given by
\begin{equation}
u \cdot v=\rho (T(u))v,\quad \forall~ u,v\in V.
\end{equation}
\end{pro}

\pf
First by  $T\circ A=\phi_\frak g\circ T$ and the fact that$(V,A,\rho)$ is a representation, we have
$$
A(u\cdot v)=A(\rho(T(u))v)=\rho (\phi_\frak g(T(u)))A(v)=\rho (T(A(u))A(v)=A(u)\cdot A(v),
$$
which implies that $A$ is an algebra homomorphism. Furthermore, we have
\begin{eqnarray*}
&&(u\cdot v)\cdot A(w)-A(u)\cdot (v\cdot w)- (v\cdot u)\cdot A(w)+A(v)\cdot (u\cdot w)\\
&=&\rho T(\rho (T(u))v) A(w)-\rho T (A(u)) \rho (T(v)) w-\rho T(\rho (T(v))u) A(w)+\rho T (A(v)) \rho (T(u)) w\\
&=&\rho T\big(\rho (T(u))v-\rho (T(v))u\big) A(w)-\rho (\phi_\frak g \circ T(u))\rho (T(v)) w+\rho (\phi_\frak g \circ T(v))\rho (T(u)) w\\
&= &\rho [T(u),T(v)]_\g A(w)- \rho (\phi_\frak g\circ  T(u))\rho (T(v)) w+\rho( \phi_\frak g\circ  T(v))\rho (T(u)) w\\
&=&0.
\end{eqnarray*}
The last equality holds also by the fact that $(V,A,\rho)$ is a representation. Therefore, $(V,\cdot, A)$ is a hom-left-symmetric algebra.  \qed\vspace{3mm}

The following consequence is obvious.

\begin{cor} With the above conditions, $T$ is a homomorphism from the commutator hom-Lie algebra $(\frak g (V), [\cdot,\cdot]_V, A)$, given in Proposition \ref{prop:left-Lie}, to $(\g,[\cdot,\cdot]_\g,\phi_\g)$. Moreover, $T(V)=\{T(u)|u\in V\}\subset \g$ is a hom-Lie subalgebra of $(\frak g,[\cdot,\cdot]_\g,\phi_\frak g)$ and there is an induced
hom-left-symmetric algebra structure $(T(V), \cdot, \phi_\frak g)$ on $T(V)$ given by
\begin{equation}
T(u)\cdot T(v)=T(u\cdot v).
\end{equation}
\end{cor}

When we take the adjoint representation, we have the following conclusion.

\begin{cor} {\rm (\cite{MY}[Theorem 4.1])} Let $(\frak g,[\cdot,\cdot]_\g,\phi_\frak g)$ be a hom-Lie algebra and $R$ be a multiplicative Rota-Baxter operator (of weight zero).
 Then $(\frak g,\cdot, \phi_\frak g)$ is a hom-left-symmetric algebra, where the multiplication $\cdot$ is defined by
$x\cdot y=[R(x), y]_\g.$
\end{cor}

\begin{cor}\label{cor:construction}
Let $(\frak g,[\cdot,\cdot]_\g,\phi_\frak g)$ be a hom-Lie algebra. Then there is a hom-left-symmetric algebra structure $(\frak g, \cdot, \phi_\frak g)$ on $\frak g$ such that the commutator hom-Lie algebra
structure is $(\frak g,[\cdot,\cdot]_\g,\phi_\frak g)$ if and only if there exists an invertible $\mathcal O$-operator $T$ associated to a representation $( V,A,\rho)$.
\end{cor}

\pf
Let $T$ be an invertible $\mathcal O$-operator associated to a representation $( V,A,\rho)$. By Proposition \ref{prop:lie-lsa}, there exists a hom-left-symmetric
algebra structure on $T(V)=\g$ given by
$$x\cdot y=T(\rho (x) T^{-1}(y)).$$
Moreover, we have
$$x\cdot y-y\cdot x= T(\rho (x) T^{-1}(y)-\rho(y) T^{-1}(x))=[T(T^{-1})(x),T(T^{-1})(y)]_\g=[x,y]_\g.$$
Conversely, the identity map ${\rm id}:\frak g\rightarrow \frak g$ is an  invertible $\mathcal O$-operator associated to the representation $(\g,\phi_\g,L)$ given in Proposition~\ref{prop:left-Lie}.
\qed

\begin{cor}
Let $(\frak g,[\cdot,\cdot]_\g,\phi_\frak g)$ be a regular  admissible  hom-Lie algebra. Suppose that $B\in\wedge^2\g^*$ is a nondegenerate
skew-symmetric bilinear form satisfying $B(\phi_\g(x),y)=B(x,\phi_\g(y))$ and Eq.~\eqref{eq:2-cocycle}. Then there is a hom-left-symmetric algebra structure $(\frak g, \cdot, \phi_\frak g)$ on $\frak g$ given by
\begin{equation}
B(x\cdot y, z)=-B({\phi_\frak g}^{-1}(y), [x,\phi_\frak g (z)]_\g),
\end{equation}                                                                                                                                                                   such that the commutator hom-Lie algebra
 is $(\frak g,[\cdot,\cdot]_\g,\phi_\frak g)$.
\end{cor}
\pf
Let $r\in \wedge^2\frak g$ be determined by $\langle {r^\sharp}^{-1}(x), y\rangle=B(x,y)$. By $B(\phi_\g(x),y)=B(x,\phi_\g(y))$, we have $r^\sharp\circ \phi_\frak g^*=\phi_\frak g\circ r^\sharp$. Therefore by the proof of Proposition \ref{pro:2-cocycle},
$\phi_\frak g\circ r^\sharp$ is an invertible $\mathcal O$-operator associated to the co-adjoint representation $(\g^*,\phi_\g^*,{\ad}^*)$. Hence by Corollary \ref{cor:construction},
 $(\frak g, \cdot, \phi_\frak g)$ is a hom-left-symmetric algebra, where the multiplication $\cdot$ is given by
\begin{eqnarray*}
B(x\cdot y, z)&=& B(\phi_\frak g\circ r^\sharp ({\rm ad}^*(x) ({r^\sharp}^{-1}\circ{\phi_\frak g}^{-1})(y)),z)\
=\langle \phi_\frak g^* ({\rm ad}^*(x) ({r^\sharp}^{-1}\circ{\phi_\frak g}^{-1})(y)), z\rangle\\
&=& \langle {\rm ad}^*(x) ({r^\sharp}^{-1}\circ{\phi_\frak g}^{-1})(y), \phi_\frak g (z)\rangle
=\langle {r^\sharp}^{-1}\circ{\phi_\frak g}^{-1}(y), -[x,\phi_\frak g(z)]_\g\rangle\\
&=&-B({\phi_\frak g}^{-1}(y), [x,\phi_\frak g (z)]_\g),
\end{eqnarray*}
such that the commutator hom-Lie algebra is $(\frak g,[\cdot,\cdot]_\g,\phi_\frak g)$.
\qed\vspace{3mm}

At the end of this section, we study a relationship between an
$\mathcal O$-operator associated to an arbitrary admissible representation and the
classical hom-Yang-Baxter equation, which leads to a construction of
solutions of the classical hom-Yang-Baxter equation in terms of $\mathcal O$-operators and
hom-left-symmetric algebras.

Let $(\frak g,[\cdot,\cdot]_\g,\phi_\frak g)$ be a hom-Lie algebra and
$(V,A,\rho)$ be an admissible representation. Then we have the semidirect product hom-Lie algebra $\g\ltimes_{\rho^*}V^*$, and  any linear map $T:V\longrightarrow\g$ can be view as an element $\overline{T}\in\otimes^2(\g\oplus V^*)$ via
$$
\overline{T}(\xi+u,\eta+v)=\langle T(u),\eta\rangle,\quad \forall~\xi+u,\eta+v\in\g^*\oplus V.
$$
Let $\sigma$ be the exchange operator acting on
 the tensor space, then $r\triangleq\overline{T}-\sigma (\overline{T})$ is skew-symmetric.

Let $\{v_1,\cdots, v_n\}$ be a basis of $V$ and
$\{v^1,\cdots,v^n\}$ be its dual basis.

\begin{thm}\label{thm:O-operator}
Let $T:V\rightarrow \frak g$ be a linear map  satisfying $T\circ A=\phi_\frak
g\circ T$. Then
 $r=\overline{T}-\sigma (\overline{T})$ is a  solution of the classical hom-Yang-Baxter equation
in the hom-Lie algebra $\frak g\ltimes_{\rho^*} V^*$ if and only if $T\circ A$ is an $\mathcal O$-operator.
Furthermore, if $A$ is orthogonal, then we have $(\phi_\g\oplus A^*)^{\otimes 2}r=r$.
\end{thm}
\pf
It is obvious that $\overline{T}$ can be expressed by
$\overline{T}= v^i\otimes T(v_i).$
Here the Einstein summation convention is used. Therefore, we can write $r=v^i\wedge T(v_i)$. By direct computations, we have
\begin{eqnarray*}
  [r,r]_{\g\ltimes V^*}&=&[v^i\wedge T(v_i),v^j\wedge T(v_j)]_{\g\ltimes V^*}\\
  &=&-[v^i,T(v_j)]_{\g\ltimes V^*}\wedge\phi_\g\circ T(v_i)\wedge A^* v^j-[T(v_i),v^j]_{\g\ltimes V^*}\wedge A^* v^i\wedge\phi_\g\circ T(v_j)\\
  &&+[T(v_i),T(v_j)]_\g\wedge A^* v^i\wedge A^* v^j\\
  &=&\langle\rho^*(T(v_j))(v^i),v_m\rangle v^m\wedge T\circ A(v_i)\wedge\langle A^*v^j,v_n\rangle v^n\\
  &&-\langle\rho^*(T(v_i))(v^j),v_n\rangle v^n\wedge\langle A^*v^i,v_m\rangle v^m\wedge T\circ A(v_j)\\
  &&+[T(v_i),T(v_j)]_\g\wedge \langle A^* v^i,v_m\rangle v^m\wedge \langle A^* v^j,v_n\rangle v^n\\
  &=&-v^m\wedge \langle v^i,\rho(T(v_j))(v_m)\rangle\langle v^j,A(v_n)\rangle T\circ A(v_i)\wedge v^n\\
  &&+v^n\wedge v^m\wedge \langle v^j,\rho(T(v_i))(v_n)\rangle\langle v^i,A(v_m)\rangle T\circ A(v_j)\\
  &&+\langle  v^i,A(v_m)\rangle\langle v^j,A(v_n)\rangle[T(v_i),T(v_j)]_\g\wedge  v^m\wedge  v^n\\
  &=&v^m\wedge v^n\wedge T\circ A(\rho(T\circ A(v_n))(v_m))+v^n\wedge v^m\wedge T\circ A(\rho(T\circ A(v_m))(v_n))\\
  &&+[T\circ A(v_m),T\circ A(v_n)]_\g\wedge  v^m\wedge  v^n\\
  &=&v^m\wedge v^n\wedge \Big( T\circ A\big(\rho(T\circ A(v_n))(v_m)-\rho(T\circ A(v_m))(v_n)\big)+[T\circ A(v_m),T\circ A(v_n)]_\g\Big).
\end{eqnarray*}
Furthermore, since $T\circ A=\phi_\frak
g\circ T$, it is obvious that $T\circ A$ satisfies $(T\circ A)\circ A=\phi_\g\circ(T\circ A) $. Thus, $r$ is a solution of the classical hom-Yang-Baxter equation
in the hom-Lie algebra $\frak g\ltimes_{\rho^*} V^*$ if and only if $T\circ A$ is an $\mathcal O$-operator.

Assume that $A=(a^i_j)$, then we have
\begin{eqnarray*}
 (\phi_\g\oplus A^*)^{\otimes 2}r&=&A^*v^i\wedge \phi_\g\circ T(v_i)=A^*v^i\wedge T(A v_i)\\
 &=&a^i_jv^j\wedge a^k_iT(v_k)=(a^i_ja^k_i)v^j\wedge T(v_k).
\end{eqnarray*}
Therefore, if $A$ is orthogonal, we have $(\phi_\g\oplus A^*)^{\otimes 2}r=r$. \qed

\begin{cor}\label{co:hlsa}
Let $(V, \cdot, \psi)$ be a regular hom-left-symmetric algebra. Suppose that
\begin{itemize}
\item[\rm(i)] $u\cdot \psi(v)=\psi^2(u)\cdot \psi(v)$;
\item[\rm (ii)] $\psi((u\cdot v-v\cdot u)\cdot w)=u\cdot (\psi(v)\cdot w)-v\cdot (\psi(u)\cdot w)$.
\end{itemize}
Then
$
r= v^i\wedge \psi^{-1}(v_i)
$
is a solution of the classical hom-Yang-Baxter equation
in the hom-Lie algebra $\frak g(A)\ltimes_{L^*} A^*$.
\end{cor}

\pf
Note that the above conditions (i) and (ii) are exactly the conditions that $(V^*,\psi^*,L^*)$ is a representation
of the commutator hom-Lie algebra $\frak g(A)$ with respect to $\psi^*$. Moreover, the identity map ${\rm id}=\psi^{-1}\circ\psi:V \rightarrow \frak g(V)$ is an  $\mathcal O$-operator associated to the representation $( V,\psi,L)$. Therefore, the conclusion follows from Theorem \ref{thm:O-operator}.
\qed

\end{document}